 \documentclass[journal,a4paper]{IEEEtran}

\usepackage[final]{fixme}

 \usepackage{cite}   
\usepackage{color}
\usepackage{psfrag}
\usepackage{pstricks,pst-node,pst-plot}
\usepackage{graphicx}
\usepackage{amsmath,amssymb,array,multirow,bigdelim}
\usepackage[amssymb]{SIunits} 
\usepackage{url}
\usepackage{t1enc}
\usepackage{url}   
\usepackage{booktabs}
\usepackage{array}

 \makeatletter
 \let\NAT@parse\undefined
 \makeatother

\renewcommand{\vec}[1]{\mathbf #1}
\newcommand{\mat}[1]{\mathbf #1}
\newcommand{\transpose}{^\mathrm{T}}

\newcommand{\CC}{\mathbb C}

\newtheorem{theorem}{Theorem}

\begin{document}

\title{ Modeling of Reverberant Radio Channels  Using Propagation Graphs}
\author{\normalsize\authorblockN{Troels Pedersen, Gerhard Steinb\"ock, and
    Bernard H. Fleury}%
  \thanks{Department of Electronic Systems, Aalborg University,
    DK-9220 Aalborg East, Denmark. Email:
    \{troels,gs,fleury\}@es.aau.dk.   } }
\maketitle
\begin{abstract}
  In measurements of in-room radio channel responses an avalanche
  effect can be observed: earliest signal components, which appear
  well separated in delay, are followed by an avalanche of components
  arriving with increasing rate of occurrence, gradually merging into
  a diffuse tail with exponentially decaying power. We model the
  channel as a propagation graph in which vertices represent
  transmitters, receivers, and scatterers, while edges represent
  propagation conditions between vertices. The recursive structure of
  the graph accounts for the exponential power decay and the avalanche
  effect.  We derive a closed form expression for the graph's transfer
  matrix. This expression is valid for any number of interactions
  and is straightforward to use in numerical simulations.  We discuss
  an example where time dispersion occurs only due to propagation in
  between vertices. Numerical experiments reveal that the graph's
  recursive structure yields both an exponential power decay and an
  avalanche effect.
\end{abstract}
\vspace{-1ex}

\section{Introduction}
Engineering of modern indoor radio systems for communications and
geolocation relies heavily on models for the time dispersion of the
wideband and ultrawideband radio channels
\cite{hashemi,Molkdar1991,Alsindi2009}.  From measurement data, as
exemplified in Fig.~\ref{fig:kunischExample}, it appears that the
spatial average of the channel impulse response's squared magnitude
(referred to as the delay-power spectrum) observed in in-room
scenarios exhibits an avalanche effect: The earliest signal
components, which appear well separated in time, are followed by an
avalanche of components arriving with increasing rate of occurrence,
gradually merging into a diffuse tail with exponentially decaying
power.
\begin{figure}
  \centering
  \includegraphics[width=0.45\linewidth]{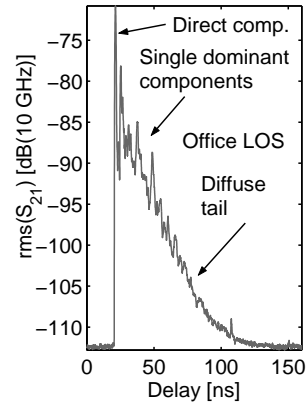}
  \caption{Spatially averaged delay-power spectrum measured within an
    office of 5$\times$5$\times$2.6\,\cubic\meter obtained as the rms
    value of impulse response for 30$\times$30 receiver positions on a
    square horizontal grid with steps of 1\,\centi\meter{}. Abscissa
    includes cable delays; reference to signal bandwidth
    (10\,\giga\hertz) shifts ordinate by --100\,\deci\bel. Reprinted
    from \cite{Kunisch2002} with permission (\copyright\ 2002 IEEE).
  }
  \label{fig:kunischExample}
\end{figure} { The diffuse tail is often referred to as a ``dense
  multipath component'' and is commonly attributed to diffuse
  scattering from rough surfaces and objects which are small compared
  to the wavelength \cite{Poutanen2011}.}
A similar avalanche effect is well-known in room acoustics
\cite{Kuttruff2000} where it is attributed to recursive scattering of
sound waves.  { Indoor radio
  propagation environments are particularly exposed to recursive
  scattering as electromagnetic waves may be reflected back and forth
  in between walls, floor, and ceiling. Thus, the avalanche effect and
  the exponential power decay may occur due to recursive scattering
  rather than diffuse scattering.}


Recursive scattering phenomena have been previously taken into account in a
number of radio channel models. The works
\cite{Holloway1999,Rudd2003,Rudd2007,Andersen2007} use the analogy to
acoustical reverberation theory to predict the exponential decay.  As
a matter of fact, there exists a well-developed theory of
electromagnetic fields in cavities \cite{Lehman1993,Hill2009}, but in
this context too the avalanche effect has received little
attention. Recursive scattering between particles in a homogeneous
medium is a well-known phenomenon studied by Foldy \cite{Foldy1945}
and Lax \cite{Lax1951,Lax1952}.  The solution, given by the so-called
Foldy-Lax equation \cite{Mishchenko2006}, has been applied in the
context of time-reversal imaging by Shi and Nehorai
\cite{Shi2007a}. The solution is, however, intractable for
heterogeneous indoor environments. In \cite{Franceschetti2004a} the
radio propagation mechanism is modeled as a ``stream of photons''
performing continuous random walks in a homogeneously cluttered
environment. The model predicts a delay-power spectrum consisting of a
single directly propagating ``coherent component'' followed by an
incoherent tail.  Time-dependent radiosity
\cite{Nosal2004,Hodgson2006,Siltanen2007,Muehleisen2009} accounting
for delay dispersion has been recently applied to design a model for
the received instantaneous power \cite{Rougeron2002}. Thereby, the
exponential power decay and the avalanche effect can be predicted.

Simulation studies of communication and localization systems commonly
rely on synthetic realizations of the channel impulse response. A
multitude of impulse response models exist
\cite{Molkdar1991,Hashemi1993,Alsindi2009}, but only few account for
the avalanche effect.  The models
\cite{Kunisch2003a,Kunisch2003,Kunisch2006} treat early components via
a geometric model whereas the diffuse tail is generated via another
stochastic process; the connection between the propagation environment
and the diffuse tail is, however, not considered. Ray tracing methods
may also be used to do site-specific simulations
\cite{Fernandez2008}. However, ray tracing methods commonly predict
the signal energy to be concentrated into the single dominant components
whereas the diffuse tail is not well represented.

{ In this contribution, we model the channel impulse response for the
  purpose of studying the avalanche effect and, in particular, its
  relation to recursive scattering. The objective is a model which
  allows for simulation of both realizations of the channel response and
  average entities such as the delay-power spectrum.  Expanding on the
  previously presented work \cite{Pedersen2007,Pedersen2006}, we
  propose a unified approach for modeling of the transition from
  single dominant components to the diffuse tail. The propagation
  environment is represented in terms of a propagation graph, with
  vertices representing transmitters, receivers, and scatterers, while
  edges represent propagation conditions between vertices.  The
  propagation graph accounts inherently for recursive scattering and
  thus signal components arriving at the receiver may have undergone
  any number of scatterer interactions.  This modelling approach
  allows for expressing the channel transfer function in closed form
  for unlimited number of interactions. The formalism enables a
  virtual ``dissection'' of the impulse response in so-called partial
  responses to inspect how recursive scattering leads to a gradual
  buildup of the diffuse tail.

  Propagation graphs may be defined according to a specific scenario
  for site-specific prediction, or be generated as outcomes of a
  stochastic process for use in e.g. Monte Carlo simulations. In the
  present contribution we consider an example of a stochastically
  generated propagation graph suitable for Monte Carlo simulations. In
  the example model, scatterer interactions are assumed to cause no
  time dispersion and thus delay dispersion occurs only due to
  propagation in between vertices.  The simulations reveal that the
  graph's recursive structure yields both an exponential power decay
  and an avalanche effect in the generated impulse responses.  }


\section{Representing Radio Channels as Graphs}
{
In a typical propagation scenario, the electromagnetic signal emitted
by a transmitter propagates through the environment while  interacting with a
number of objects called scatterers.  The receiver, which is usually
placed away from the transmitter, senses the electromagnetic
signal. If a line of sight exists between the transmitter and
the receiver, direct propagation occurs. Also, the signal may arrive at
the receiver indirectly via one or more scatterers. 
In the following we represent propagation
mechanisms using graphs allowing for both recursive and non-recursive scattering. 
We first introduce the necessary definitions of directed graphs.
}

\subsection{Directed Graphs}
\label{subsec:directed-graphs}
{
Following \cite{Diestel2000}, we define a directed graph $\mathcal{G}$
as a pair $(\mathcal{V},\mathcal{E})$ of disjoint sets of vertices and
edges. Edge $e\in \mathcal E$ with initial vertex denoted by
$\mathrm{init}(e)$ and terminal vertex denoted by $\mathrm{term}(e)$
is said to be outgoing from  vertex
$\mathrm{init}(e)$ and ingoing to $\mathrm{term}(e)$.  We consider
graphs without parallel edges. Thus, there exists no pair of edges
$e$ and $e'$ such that $\mathrm{init}(e) = \mathrm{init}(e')$ and
$\mathrm{term}(e)=\mathrm{term}(e')$. In this case an edge $e$ can be
identified by the vertex pair $(\mathrm{init}(e),\mathrm{term}(e))\in
\mathcal{V}^2$ and with a slight abuse of notation we write
$e=(\mathrm{init}(e),\mathrm{term}(e))$ and $\mathcal{E}\subseteq
\mathcal{V}\times\mathcal{V} $. 
}
\subsection{Propagation Graphs}
\label{subsec:propagation-graphs}

We define a propagation graph as a directed graph where vertices
represent transmitters, receivers and scatterers.  Edges represent the
propagation conditions between the vertices. Thus, the vertex set
$\mathcal V$ of a propagation graph is a union of three disjoint sets:
$\mathcal V = \mathcal V_{\mathrm t} \cup \mathcal V_{\mathrm r} \cup
\mathcal V_{\mathrm s}$, the set of transmitters $\mathcal V_{\mathrm
  t}=\{\mathrm{Tx1},\dots,\mathrm{Tx}N_{\mathrm t} \}$, the set of receivers
$\mathcal V_{\mathrm r}=\{\mathrm{Rx1},\dots,\mathrm{Rx}N_{\mathrm r}
\}$, and the set of scatterers $\mathcal V_{\mathrm
  s}=\{\mathrm{S1},\dots,\mathrm{S}N_{\mathrm s}\}$.  The transmit vertices are
considered as sources and have only outgoing edges. Likewise, the
receivers are considered as sinks with only incoming edges. The edge
set $\mathcal E$ can thus be partitioned into
four disjunct sets as $\mathcal E = \mathcal E_{\mathrm d} \cup
\mathcal E_{\mathrm t} \cup \mathcal E_{\mathrm r} \cup \mathcal
E_{\mathrm s},$ with direct edges in $ \mathcal E_{\mathrm d} =\mathcal E \cap
(\mathcal V_{\mathrm t}\times \mathcal V_{\mathrm r} )$,
transmitter-scatterer edges in $ \mathcal E_{\mathrm t} =\mathcal E \cap (\mathcal
V_{\mathrm t}\times \mathcal V_{\mathrm s} ) $, 
scatterer-receiver edges in
 $ \mathcal E_{\mathrm r} =\mathcal E \cap
(\mathcal V_{\mathrm s}\times \mathcal V_{\mathrm r} ) $, and
 inter-scatterer edges in
 $ \mathcal E_{\mathrm s} =\mathcal E
\cap (\mathcal V_{\mathrm s}\times \mathcal V_{\mathrm s} )$. Fig.~\ref{fig:propagationgraph:mimo}
shows an example propagation graph.


\begin{figure}
\centering
\resizebox{0.8\columnwidth}{!}{
\psset{unit=0.55cm}
\begin{pspicture}(-1.9,-3)(13,11)
\psset{gridcolor=yellow}
\psset{subgridcolor=yellow}
\psset{linewidth=0.5pt}


\rput(0,2){
\rput[c]{0}(0,0){\dotnode(0,0){Tx4}\nput[labelsep=0.1]{180}{Tx4}{Tx4}}
\rput[c]{0}(0,1){\dotnode(0,0){Tx3}\nput[labelsep=0.1]{180}{Tx3}{Tx3}}
\rput[c]{0}(0,2){\dotnode(0,0){Tx2}\nput[labelsep=0.1]{180}{Tx2}{Tx2}}
\rput[c]{0}(0,3){\dotnode(0,0){Tx1}\nput[labelsep=0.1]{180}{Tx1}{Tx1}}
\psccurve[linecolor=gray](0,-1)(1,4)(-1,4)(-1.4,0)
\rput*(-1,4){$\mathcal V_{\mathrm{t}}$}
}

\rput(5,7){
\rput[c]{0}(0,2){\dotnode(0,0){Rx1}\nput[labelsep=0.1]{90}{Rx1}{Rx1}}
\rput[c]{0}(1,1){\dotnode(0,0){Rx2}\nput[labelsep=0.1]{90}{Rx2}{Rx2}}
\rput[c]{0}(0,0){\dotnode(0,0){Rx3}\nput[labelsep=0.1]{90}{Rx3}{Rx3}}
\psccurve[linecolor=gray](0,-1)(2,0)(2,3)(0,3)(-1,2.2)
\rput*(2,3){$\mathcal V_{\mathrm{r}}$}
}

\rput(6,0){
  \rput[c]{0}(1,2){\dotnode(0,0){s1}\nput[labelsep=0.2]{-90}{s1}{S1}}
  \rput[c]{0}(3,1){\dotnode(0,0){s2}\nput[labelsep=0.1]{-90}{s2}{S2}}
  \rput[c]{0}(3.5,3.3){\dotnode(0,0){s3}\nput[labelsep=0.15]{60}{s3}{S3}}
  \rput[c]{0}(5,2){\dotnode(0,0){s4}\nput[labelsep=0.1]{-45}{s4}{S4}}
  \rput[c]{0}(6,3){\dotnode(0,0){s5}\nput[labelsep=0.1]{90}{s5}{S5}}
  \rput[c]{0}(3,5){\dotnode(0,0){s6}\nput[labelsep=0.1]{0}{s6}{S6}}
  \psccurve[linecolor=gray](0,0)(1,1)(2,-0.2)(3,0.2)(4,0)(7,2)(7,4)(6,6)(5,5.5)(4,6)(1,5)
\rput*(7,4){$\mathcal V_{\mathrm{s}}$}
}

\ncline[nodesepA=0.1, nodesepB=0.1]{cc->}{Tx1}{Rx1}\lput{:U}{}
\ncline[nodesepA=0.1, nodesepB=0.1]{cc->}{Tx2}{Rx1}\lput{:U}{}
\ncline[nodesepA=0.1, nodesepB=0.1]{cc->}{Tx2}{s3}\lput{:U}{}
\ncline[nodesepA=0.1, nodesepB=0.1]{cc->}{Tx3}{Rx1}\lput{:U}{}
\ncline[nodesepA=0.1, nodesepB=0.1]{cc->}{Tx3}{s3}\lput{:U}{}

\ncline[nodesepA=0.1, nodesepB=0.1]{cc->}{Tx4}{s6}\lput{:U}{}
\ncline[nodesepA=0.1, nodesepB=0.1]{cc->}{Tx4}{s1}\lput{:U}{}
\ncline[nodesepA=0.1, nodesepB=0.1]{cc->}{s3}{s1}\lput{:U}{}
\ncarc[nodesepA=0.1, nodesepB=0.1]{cc->}{s1}{s2}\lput{:U}{}
\ncarc[nodesepA=0.1, nodesepB=0.1]{cc->}{s2}{s1}\lput{:U}{}
\ncline[nodesepA=0.1, nodesepB=0.1]{cc->}{s2}{s4}\lput{:U}{}
\ncline[nodesepA=0.1, nodesepB=0.1]{cc->}{s4}{s3}\lput{:U}{}
\ncline[nodesepA=0.1, nodesepB=0.1]{cc->}{s4}{s5}\lput{:U}{}
\ncline[nodesepA=0.1, nodesepB=0.1]{cc->}{s6}{Rx2}\lput{:U}{}
\ncline[nodesepA=0.1, nodesepB=0.1]{cc->}{s3}{Rx3}\lput{:U}{}
\ncline[nodesepA=0.1, nodesepB=0.1]{cc->}{s1}{Rx3}\lput{:U}{}

\rput[tl](-2.5,-0.1){
\begin{minipage}{0.8\linewidth}
\scriptsize
\vspace{-2ex}
\begin{align*}
\mathcal E_{\mathrm d} &= \{
(\mathrm{Tx1},\mathrm{Rx1}),\allowbreak(\mathrm{Tx2},\mathrm{Rx1}),\allowbreak(\mathrm{Tx3},\mathrm{Rx1})\} 
\\
\mathcal E_{\mathrm t} &= \{
(\mathrm{Tx2},\mathrm{S3}),\allowbreak
(\mathrm{Tx3},\mathrm{S3}),\allowbreak
(\mathrm{Tx4},\mathrm{S6}),\allowbreak
(\mathrm{Tx4},\mathrm{S1})
\}
\\
\mathcal E_{\mathrm r} &= \{
(\mathrm{S1},\mathrm{Rx3}),\allowbreak
(\mathrm{S3},\mathrm{Rx3}),\allowbreak
(\mathrm{S6},\mathrm{Rx2})
\}\\
\mathcal E_{\mathrm s} &= \{
(\mathrm{S1},\mathrm{S2}),\allowbreak
(\mathrm{S2},\mathrm{S1}),\allowbreak
(\mathrm{S3},\mathrm{S1}),\allowbreak
(\mathrm{S2},\mathrm{S4}),\allowbreak
(\mathrm{S4},\mathrm{S3}),\allowbreak
(\mathrm{S4},\mathrm{S5})
\}   
\end{align*}
\end{minipage}
}
\end{pspicture}
} 
\caption{A propagation graph with four transmit vertices, three receive
  vertices and six scatterer vertices. }
  \label{fig:propagationgraph:mimo}
\end{figure}
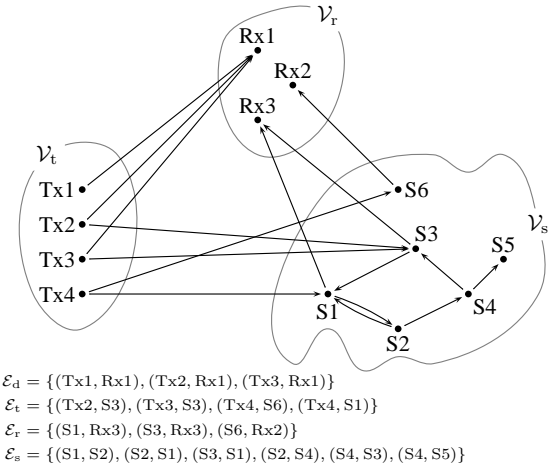



The signals propagate in the graph in the following way. Each
transmitter emits a signal that propagates via its outgoing edges. The
signal observed by a receiver vertex is the sum of the signals
arriving via the ingoing edges. A scatterer sums up the signals on its
ingoing edges and re-emits the sum-signal on the outgoing edges.  As a
signal propagates along an edge, or interacts with a scatterer, it
undergoes delay and dispersion in time. The specific delay dispersion
endured by a signal depends on the particular propagation mechanism
along its edges. Assuming these mechanisms to be linear and
time-invariant, this effect can be represented as a convolution with
an impulse response or, in the Fourier domain, as a multiplication
with a transfer function.  Hence, the signal arriving to vertex
$v_{n'}$ via edge $e=(v_n,v_{n'})$ reads $A_{e}(f)C_n(f)$, where
$C_n(f)$ is the signal emitted by vertex $v_n$ and $A_e(f)$ denotes
the transfer function associated to edge $e$. In other words, the
transfer function $A_{e}(f)$ describes the interaction at the initial
vertex $v_n$ and the propagation from $v_n$ to $v_{n'}$.

\subsection{Weighted Adjacency Matrix of a Propagation Graph }
Propagation along the edges is described via a transfer matrix $\mat
A(f)$ which can be viewed as an adjacency matrix with each entry
coinciding with the edge transfer function of the corresponding
edge. Thus, the weighted adjacency matrix $\mat A(f) \in
\CC^{(N_{\mathrm t}+N_{\mathrm r}+N_{\mathrm s})\times(N_{\mathrm
    t}+N_{\mathrm r}+N_{\mathrm s})}$ of the propagation graph
$\mathcal G$ is defined as
 \begin{equation}
  \label{eq:15}
  [\mat A(f)]_{nn'} =
  \begin{cases}
    A_{(v_n,v_{n'})}(f) & \text{if } (v_n,v_{n'})\in \mathcal E,\\
    0 & \text{otherwise,}
  \end{cases}
\end{equation}
i.e., entry $n,n'$ of $\mat A(f)$ is the transfer function from vertex
$v_n$ to vertex $v_{n'}$ of $\mathcal G$.  
Selecting the indexing of the vertices according to
\begin{equation}
  \label{eq:11}
  v_n\in
\begin{cases}
    \mathcal V_{\mathrm{t}}, &n=1,\dots,N_{\mathrm t} \\
    \mathcal V_{\mathrm{r}}, &n=N_{\mathrm t} + 1,\dots,N_{\mathrm t}+N_{\mathrm r} \\
    \mathcal V_{\mathrm{s}}, &n=N_{\mathrm t}+N_{\mathrm r}+1,\dots,N_{\mathrm t} +N_{\mathrm r}+N_{\mathrm s},
  \end{cases}
\end{equation}
the weighted adjacency matrix takes the form 
\begin{equation}
  \label{eq:44}
   \mat A(f) = 
   \begin{bmatrix}
     \mat 0 &\mat 0&  \mat 0 \\ 
     \mat D(f) & \mat 0 & \mat R(f) \\ 
     \mat T(f) &  \mat 0 &  \mat B(f) 
   \end{bmatrix},
\end{equation}
where $\mat 0$ denotes the all-zero matrix of the appropriate
dimension and the transfer matrices
\begin{alignat}{4} 
  \label{eq:45a}
     \mat D(f)& \in \CC^{N_{\mathrm r}\times N_{\mathrm t}}\quad&&\text{connecting
       $\mathcal V_\mathrm t$ to $\mathcal V_\mathrm r$,}\\
\label{eq:45b}
     \mat R(f)& \in \CC^{N_{\mathrm r}\times N_{\mathrm s}} \quad &&\text{connecting
       $\mathcal V_\mathrm s$ to $\mathcal V_\mathrm r$,}\\ 
\label{eq:45c}
     \mat T(f)& \in \CC^{N_{\mathrm s}\times N_{\mathrm t}} \quad&& \text{connecting
       $\mathcal V_\mathrm t$ to $\mathcal V_\mathrm s$, and} \\
\label{eq:45d}
     \mat B(f)& \in \CC^{N_{\mathrm s}\times N_{\mathrm s}}\quad &&\text{interconnecting
       $\mathcal V_\mathrm s$.}
\end{alignat}
The special structure of $\mat A(f)$ reflects the structure of the
propagation graph. The first $N_{\mathrm t}$ rows are zero because we do not
accept incoming edges into the transmitters. Likewise columns
$N_{\mathrm t}+1,\dots,N_{\mathrm t}+N_{\mathrm r}$ are all zero since the receiver vertices have no
outgoing edges.

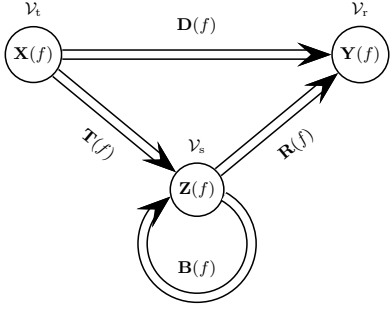
\begin{figure}
  \centering
\resizebox{0.9\columnwidth}{!}{
\begin{pspicture}(-1,-2)(10,4.5)
\small
\cnodeput(1.5,3){vt}{$\vec X(f)$}
\cnodeput(7.5,3){vr}{$\vec Y(f)$}
\cnodeput(4.5,0.5){vs}{$\vec Z(f)$}
\nput{90}{vt}{$\mathcal V_{\mathrm{t}}$}
\nput{90}{vr}{$\mathcal V_{\mathrm{r}}$}
\nput{90}{vs}{$\mathcal V_{\mathrm{s}}$}
\ncline[doubleline=true,doublesep=5\pslinewidth,arrows=->]{vt}{vr}
\naput[nrot=:U]{$\mat D(f)$}
\ncline[doubleline=true,doublesep=5\pslinewidth,arrows=->]{vt}{vs}
\nbput[nrot=:U]{$\mat T(f)$}
\ncline[doubleline=true,doublesep=5\pslinewidth,arrows=->]{vs}{vr}
\nbput[nrot=:U]{$\mat R(f)$}
\nccircle[angleA=-180,doubleline=true,doublesep=5\pslinewidth,arrows=<-]{vs}{1}
\naput{$\mat B(f)$}
\end{pspicture}}
  \caption{Vector signal flow graph representation of a propagation
    graph. Vertices represent vertex sets of the propagation graph
    with associated vector signals. Signal transmission between the
    sets are represented by the edges and associated transfer matrices. }
  \label{fig:blockDiagram}
\end{figure}

The input signal vector $\vec X(f)$ is defined as
\begin{equation}
  \label{eq:7}
\vec X(f)= 
   [ X_{1}(f),\dots,X_{N_{\mathrm t}}(f)]\transpose,
\end{equation}
where $X_m(f)$ is the spectrum of the signal emitted by transmitter
$\mathrm{Tx}m$, and $[\cdot]\transpose$ denotes the transposition
operator.  The output signal vector $\vec Y(f)$ is defined as
\begin{equation}
  \label{eq:8}
  \vec Y(f) =[Y_1(f),\dots, Y_{N_{\mathrm r}}(f)]\transpose,
\end{equation}
where $Y_m(f)$ is the Fourier transform of the signal observed by
receiver $\mathrm{Rx}m$. Similar, to $\vec X(f)$ and $\vec Y(f)$ we
let $\vec Z(f)$ denote the output signal vector of the scatterers:
\begin{align}
\label{eq:9}
\vec Z(f) =[
    Z_{1}(f),\dots,Z_{N_{\mathrm s}}(f)]\transpose
\end{align}
with the $n$th entry denoting the Fourier transform of the signal
observed at scatterer vertex $\mathrm{S}n$. 
By the definition of the propagation graph, there are no other signal
sources than the vertices in $\mathcal V_{\mathrm{t}}$.
Assuming linear and time-invariant propagation mechanisms, the
input-output relation in the Fourier domain reads
\begin{equation}
  \label{eq:10}
  \vec Y(f) = \mat H(f) \vec X(f),
\end{equation}
where $\mat H(f)$ is the $N_{\mathrm r}\times N_{\mathrm t}$ transfer
matrix of the propagation graph.  

The structure of the propagation graph unfolds in the vector signal
flow graph depicted in Fig.~\ref{fig:blockDiagram}. The vertices of
the vector signal flow graph represent the three sets $\mathcal
V_\mathrm t$, $\mathcal V_\mathrm r$, and $\mathcal V_\mathrm s$ with
the associated signals $\vec X(f)$,$\vec Y(f)$, and $\vec Z(f)$. The
edge transfer matrices of the vector signal flow graph are the
sub-matrices of $\mat A(f)$  defined in
\eqref{eq:45a}--\eqref{eq:45d}.

\section{Transfer Matrix of a  Propagation Graph}
\label{sec:gener-transf-funct}
In the following we derive the transfer matrix of a propagation
graph.  In Subsection~\ref{sec:prop-paths-walks} we first discuss how
the response of a graph is composed of signal contributions
propagating via different propagation paths. This representation is,
albeit intuitive, impractical for computation of the transfer matrix
of graphs with cycles.  Thus in
Subsections~\ref{sec:transf-funct-recurs} and
\ref{sec:part-transf-funct} we give the transfer matrix and partial
transfer matrices of a general propagation graph in closed form. Subsection
\ref{sec:recipr-prop-graphs} treats the graphical interpretation of
reciprocal channels. The section concludes with a discussion of 
related results in the literature.

\subsection{Propagation Paths and Walks }
\label{sec:prop-paths-walks}
The concept of a propagation path is a corner stone in modeling
multipath propagation. In the literature, this concept is most often
defined in terms of the resulting signal components arriving at the
receiver. A shortcoming of this definition is that it is often hard to
relate to the propagation environment. The graph terminology offers a
convenient alternative.  A walk $\ell$ (of length $K$) in a graph
$\mathcal{G} = (\mathcal V, \mathcal E)$ is defined as a sequence
$\ell = ( v^{(1)},v^{(2)},\dots, v^{(K+1)})$ of vertices in
$\mathcal{V}$ such that $(v^{(k)},v^{(k+1)}) \in \mathcal E, k=1,\dots
K$. We say that $\ell$ is a walk from $v^{(1)}$ to $v^{(K+1)}$ where
$v^{(1)}$ is the start vertex and $v^{(K+1)}$ is the end vertex;
$v^{(2)},\dots,v^{(K)}$ are called the inner vertices of $\ell$.
We define a propagation path as a walk in a propagation graph from a
transmitter to a receiver. Consequently, all (if any) inner vertices
of a propagation path are scatters.  A signal that propagates along
propagation path $\ell$ traverses $K+1$ (not necessarily different)
edges and undergoes $K$ interactions. We refer to such a propagation
path as a $K$-bounce path. The zero-bounce path $\ell=( v, v')$ is
called the direct path, from transmitter $v$ to receiver $v'$. As an
example, referring to the graph depicted in
Fig.~\ref{fig:propagationgraph:mimo}, it is straightforward to verify
that $\ell_1 = (\mathrm{Tx1},\mathrm{Rx1}) $ is a direct path, $\ell_2
= (\mathrm{Tx4},\mathrm{S6},\mathrm{Rx2}) $ is a single-bounce path,
and $\ell_3 =
(\mathrm{Tx4},\mathrm{S1},\mathrm{S2},\mathrm{S1},\mathrm{Rx3}) $ is a
3-bounce path.

We denote by $\mathcal{L}_{vv'}$ the set of propagation
paths in $\mathcal{G}$ from transmitter $v$ to receiver $v'$.  The
signal received at $v'$ originating from transmitter $v$ is the
superposition of signal components each propagating via a propagation
path in $\mathcal{L}_{vv'}$.  Correspondingly, entry $(v,v')$ of $\mat
H(f)$ reads
\begin{equation}
  \label{eq:4}
  H_{vv'}(f) = \sum_{\ell\in\mathcal{L}_{vv'}} H_\ell(f),
\end{equation}
where $H_\ell(f)$ is the transfer function of propagation path $\ell$. 

The number of terms in \eqref{eq:4} equals the cardinality of
$\mathcal{L}_{vv'}$, which may, depending on the structure of the
graph, be finite or infinite. As an example, the number of propagation
paths is infinite if $v$ and $v'$ are connected via a directed cycle
in $\mathcal G$, i.e. if a propagation path contains a walk from an
inner vertex back to itself. The graph in
Fig.~\ref{fig:propagationgraph:mimo} contains two directed cycles
which are connected to both transmitters and receivers.

In the case of an infinite number of propagation paths, computing
$H_{vv'}(f)$ directly from \eqref{eq:4} is infeasible. This problem is
commonly circumvented by truncating the sum in \eqref{eq:4} to
approximate $H_{vv'}(f)$ as a finite sum. This approach, however,
calls for a method for determining how many terms of the sum should be
included in order to achieve reasonable approximation.

In the frequently used ``$K$-bounce channel models'', propagation
paths with more than  $K$ interactions are ignored. This approach is
motivated by the rationale that at each interaction, the signal is attenuated,
and thus terms in \eqref{eq:4} resulting from propagation paths with a
large number of bounces are weak and can be left out as they do not
affect the sum much. This reasoning, however, holds true only if the
\emph{sum} of the components with more than $K$ interactions is
insignificant, which may or may not be the case. From this
consideration, it is clear that the truncation criterion is
non-trivial as it essentially necessitates computation of the whole
sum before deciding whether a term can be ignored or not.

\subsection{Transfer Matrix for Recursive and Non-Recursive
  Propagation Graphs}
\label{sec:transf-funct-recurs}
As an alternative to the approximation methods applied to the sum
\eqref{eq:4} we now give an exact closed-form expression for the
transfer matrix $\mat H(f)$. Provided
that the spectral radius of $\mat B(f)$ is less
than unity, the expression holds true for any number of terms in the sum
\eqref{eq:4} and thus holds regardless whether the number of
propagation paths is finite or infinite.
\begin{theorem}
\label{thm:transf-funct-recurs}
If the spectral radius of $\mat B(f)$  is less than unity,
then the transfer matrix of a propagation graph reads
\begin{align}
  \label{eq:60}
  \mat H(f) = \mat D(f) + \mat R(f)[\mat I - \mat B(f)]^{-1} \mat
  T(f).
\end{align}
\vspace{-2ex}
\end{theorem}

According to Theorem~\ref{thm:transf-funct-recurs} the transfer matrix
$\vec H(f)$ consists of the two following terms: $\mat D(f)$
representing direct propagation between the transmitters and receivers
and $\mat R(f) [\mat I - \mat B(f)]^{-1} \mat T(f)$ describing
indirect propagation.  The condition that the spectral radius
of $\mat B(f)$ be less than unity implies that for any vector norm
$\|\cdot \|$, $\|\mat Z(f)\| > \| \mat B(f)\vec Z(f) \|$ for non-zero
$\|\vec Z(f)\|$, cf. \cite{Horn1985}. For the Euclidean norm in
particular this condition implies the sensible physical requirement
that the signal power strictly  decreases for each interaction.

\begin{proof}
Let $\mat H_k(f)$ denote the transfer matrix for all $k$-bounce
propagation paths, then $\mat H(f)$ can be decomposed as
\begin{equation}
  \label{eq:22}
\mat H(f) = \sum_{k=0}^{\infty} \mat H_k(f),
\end{equation}
where
\begin{equation}
   \label{eq:29}
  \mat H_k(f) =
  \begin{cases}
    \mat D(f), & k = 0\\
    \mat R(f)\mat B^{k-1}(f)\mat T(f),& k>0.
  \end{cases}
\end{equation}
Insertion of \eqref{eq:29} into \eqref{eq:22} yields
\begin{align}
  \mat H(f) &= \mat D(f) + \mat R(f) \left(\sum_{k=1}^{\infty}\mat
  B^{k-1}(f)\right)  \mat T(f).
  \label{eq:25}
\end{align}
The infinite sum in \eqref{eq:25} is a Neumann series converging to
$[\mat I - \mat B(f)]^{-1}$ if the spectral radius of $\mat B(f)$) is
less than unity.  Inserting this in \eqref{eq:25}
completes the proof.
\end{proof}

The decomposition introduced in \eqref{eq:22} makes the effect of the
recursive scattering directly visible. The received signal vector is a
sum of infinitely many components resulting from any number of
interactions. The structure of the propagation mechanism is further
exposed by \eqref{eq:25} where the emitted vector signal is
re-scattered successively in the propagation environment leading to
the observed Neumann series. This allows for modeling of channels with
infinite impulse responses by expression \eqref{eq:60}. 

It is possible to arrive at \eqref{eq:60} in an alternative, but less
explicit, manner:

 \begin{proof}
   It is readily observed from the vector signal flow graph in
   Fig.~\ref{fig:blockDiagram} that $\vec Z(f)$ can be expressed as
\begin{align}
  \label{eq:27}
  \vec Z(f) = \mat T(f) \vec X(f) + \mat B(f)   \vec Z(f).
\end{align}
Since the spectral radius of $\mat B(f)$ is less than unity we obtain
for $\mat Z(f)$ the solution
\begin{equation}
  \label{eq:56}
    \vec Z(f) = [\mat I - \mat B(f)]^{-1} \mat T(f) \vec X(f).
\end{equation}
Furthermore, according to Fig.~\ref{fig:blockDiagram} the received signal is
of the form
\begin{align}
  \label{eq:57}  
  \vec Y(f) &= \mat D(f) \vec X(f) + \mat R(f)   \vec Z(f).
\end{align}
Insertion of \eqref{eq:56} in this expression yields
\eqref{eq:60}. 
\end{proof}

We remark that the above two proofs allow for propagation paths with any number of
bounces. This is highly preferable, as the derived expression
\eqref{eq:60} is not impaired by approximation errors due to the
truncation of the series into a finite number of terms as it occurs
when using $K$-bounce models.

A significant virtue of the expression \eqref{eq:60} is that
propagation effects related to the transmitters and receivers are
accounted for in the matrices $\mat D(f)$, $\mat T(f)$ and $\mat
R(f)$, but do not affect $\mat B(f)$.  Consequently, the matrix
$[\mat I-\mat B(f)]^{-1}$ only needs to be computed once even though
the configuration of transmitters and receivers changes. This is
especially advantageous for simulation studies of e.g. spatial
correlation as this leads to a significant reduction in computational
complexity.

\subsection{Partial Transfer Matrices}
\label{sec:part-transf-funct}
The closed form expression \eqref{eq:60} for the transfer matrix of a
propagation graph accounts for propagation via an arbitrary number of
scatterer interactions. For some applications it is, however, relevant
to study only some part of the impulse response associated with a
particular number of interactions. One case is where a propagation
graph is used to generate only a part of the response and other
techniques are used for the remaining parts. Another case is when one
must assess the approximation error when the infinite series is
truncated.  In the following we derive a few useful expressions for
such partial transfer matrices.

We define the $K:L$ partial transfer matrix as
\begin{equation}
  \label{eq:31}
  \mat H_{K:L}(f) = \sum_{k=K}^L \mat H_{k}(f), \quad  0\leq K\leq L,
\end{equation}
i.e., we include only contributions from propagation paths with at least $K$,
but no more than $L$ bounces. It is straightforward to evaluate
\eqref{eq:31} for $K=0$, and $L=0,1,2$:
\begin{align}
  \label{eq:24}
 \mat   H_{0:0}(f) & = \mat D(f)\\
\label{eq:51}
 \mat   H_{0:1}(f) & = \mat D(f)+ \mat R(f) \mat T (f) \\ 
\label{eq:53}
\mat   H_{0:2}(f) & = \mat D(f)+ \mat R(f) \mat T (f) 
+\mat R(f)\mat B(f) \mat T (f).
\end{align}
This expansion of the truncated series is quite intuitive but the
obtained expressions are increasingly complex for large
$L$. Theorem~\ref{thm:part-transf-matr} gives a closed form
expression of the partial transfer function $\mat H_{K,L}(f)$ for
arbitrary $K$ and $L$:

\begin{theorem}
\label{thm:part-transf-matr}
  The partial response $\mat H_{K:L}(f)$  is given by
  \begin{align}
    \mat H_{K:L}(f) &= 
\begin{cases}
\mat D(f) + \mat R(f)[\mat
I - \mat B^{L}(f)][\mat I -\mat B(f)]^{-1}\mat
T(f), & \\  \hfill K = 0, L\geq 0 & \\[1ex]
\mat R(f)[\mat
B^{K-1}(f) - \mat B^{L}(f)][\mat I -\mat B(f)]^{-1}\mat
T(f), & \\  \hfill   0<K\leq L, & \\[1ex]
\end{cases}\nonumber
  \end{align}
provided that the spectral radius of  $\mat B(f)$ is less than unity.
\end{theorem}

\begin{proof}
The partial transfer function for $ 0\leq K\leq L$ reads 
\begin{align}
  \mat H_{K:L}(f)
&= \sum_{k=K}^\infty \mat H_{k}(f) - \sum_{k'=L+1}^\infty \mat H_{k'}(f) \nonumber \\
&=\mat H_{K:\infty}(f) - \mat H_{L+1:\infty}(f).
  \label{eq:52}
\end{align}
For $K=0$ we have $\mat H_{0:\infty}(f) = \mat H(f)$ by definition;
for $K\geq 1$ we have
\begin{align}
  \mat H_{K:\infty}(f)
&= \mat R(f) \displaystyle \sum_{k=K-1}^\infty \mat B^k(f) \mat T(f)
\nonumber\\
&= \mat R(f) \mat B^{K-1}(f)\sum_{k=0}^\infty \mat B^k(f) \mat T(f)\nonumber \\
&= \mat R(f) \mat B^{K-1}(f)[\mat I - \mat B(f)]^{-1}\mat T(f).
\label{eq:12}
\end{align}
Inserting \eqref{eq:12} into   \eqref{eq:52} completes the proof.
\end{proof}

Theorem~\ref{thm:part-transf-matr} enables closed-form computation of
$\mat H_{K:L}(f)$ for any $K\geq L$.  We have already listed a few
partial transfer matrices in \eqref{eq:24}, \eqref{eq:51}, and
\eqref{eq:53}.  By definition the partial response $\mat H_{K:K}(f)$
equals $\mat H_{K}(f)$ for which an expression is provided in
\eqref{eq:29}.  The transfer function of the $K$-bounce approximation
is equal to $\mat H_{0:K}(f)$.  Another special case worth mentioning is
$\mat H_{K+1:\infty}(f) = \mat H(f) - \mat H_{0:K}(f)$ available from
\eqref{eq:12}, which gives the error due to the $K$-bounce
approximation. Thus the validity of the $K$-bounce approximation can
be assessed by evaluating some appropriate norm of $\mat
H_{K+1:\infty}(f)$.

\subsection{Reciprocity and Reverse Graphs}
\label{sec:recipr-prop-graphs}
In most cases, the radio channel is considered reciprocal.  As we
shall see shortly, the graph terminology accommodates an interesting
interpretation of the concept of reciprocity. For any propagation
graph we can define the reverse graph in which the roles of
transmitter and receiver vertices are swapped. The principle of
reciprocity states that the transfer matrix of the reverse channel is
equal to the transposed transfer matrix of the forward channel, i.e.,
a forward channel with transfer matrix $\mat H(f)$ has a reverse
channel with transfer matrix $\tilde{\mat H}(f) = \mat
H\transpose(f)$.  In the sequel we mark all entities related to the
reverse channel with a tilde.

We seek the relation between the forward graph $\mathcal G = (\mathcal
V, \mathcal E)$ and its reverse $\tilde{\mathcal G} = (\tilde
{\mathcal V}, \tilde {\mathcal E})$ under the assumption of
reciprocity. More specifically, we are interested in the relation
between the weighted adjacency matrix $\mat A(f)$ of $\mathcal G$ and
the weighted adjacency matrix $\tilde{\mat A}(f)$ of $\tilde{\mathcal
  G}$. We shall prove the following theorem:
\begin{theorem}
\label{thm:recipr-prop-graphs}
For a propagation graph $\mathcal G=(\mathcal V,\mathcal E)$ with
weighted adjacency matrix $\mat A(f)$ and transfer matrix $\mat H(f)$
there exists a reverse graph $ \tilde{\mathcal G} = (\tilde{\mathcal
  V}, \tilde{\mathcal E})$ such that $\tilde{\mathcal V}= \mathcal V$,
$ \tilde{\mathcal E} = \{(v,v'): (v',v)\in \mathcal E\} $,
$\tilde{\mat A}(f) = \mat A\transpose(f)$, and $\tilde{\mat H}(f) =
\mat H\transpose (f) $.
\end{theorem}

\begin{proof}
  We prove the Theorem by constructing a suitable propagation graph
  such that the reciprocity condition is fulfilled. We first note that
  the set of transmitters, receivers, and scatterers is maintained for
  the reverse channel, thus the vertex set of $\tilde{\mathcal G}$ is
  $ \mathcal V$ with $\tilde{\mathcal V}_{\mathrm t} = \mathcal
  V_{\mathrm r}$, $\tilde{\mathcal V}_{\mathrm r} = \mathcal
  V_{\mathrm t}$, and $\tilde{\mathcal V}_{\mathrm s} = \mathcal
  V_{\mathrm s}$.  It is immediately clear from the structure of
  $\mathcal G$ that the reverse graph $\tilde{\mathcal G}$ has no
  ingoing edges to $\tilde{\mathcal V}_{\mathrm t}$ and no outgoing
  edges from $\tilde{\mathcal V}_{\mathrm r} $ and $\tilde{\mathcal
    G}$ is thus a propagation graph. Assuming the vertex indexing as
  in \eqref{eq:11}, the weighted adjacency matrix of $\tilde{\mathcal
    G}$ is of the form
\begin{equation}
\label{eq:36}
\tilde{\mat A}(f) =   \begin{bmatrix}
    \mat 0 & \tilde {\mat D}(f) &\tilde {\mat T}(f)\\
    \mat 0 & \mat 0 &\mat 0\\
    \mat 0 & \tilde {\mat R}(f)  &\tilde {\mat B}(f)
  \end{bmatrix} ,
\end{equation}
where the transfer matrices $\tilde{\mat D}(f),\tilde{\mat
  R}(f),\tilde{\mat T}(f),$ and $\tilde{\mat B}(f)$ are defined
according to Fig.~\ref{fig:blockDiagramReverse}. Equating $\tilde{\mat
  A}(f) $ and $\mat A(f)$ we obtain the identities $\tilde {\mat D}(f) = \mat
D\transpose (f), \tilde {\mat B}(f) = \mat B\transpose (f), \tilde
{\mat T}(f) = \mat R\transpose (f), $ and $ \tilde {\mat R}(f) = \mat
T\transpose (f)$. The relation between $\tilde{\mathcal E}$ and
$\mathcal E$ now follows from the definition of the weighted adjacency
matrix. The input-output relation of the reverse channel
reads $ \tilde {\vec Y}(f) = \tilde {\mat H}(f) \tilde {\vec X}(f) $
where $\tilde {\vec X}(f)$ is the signal vector emitted by the
vertices in $\tilde{\mathcal V}_{\mathrm t}$ and $\tilde {\vec Y}(f)$
is the signal vector received by the vertices in $\tilde{\mathcal
  V}_{\mathrm t}$. 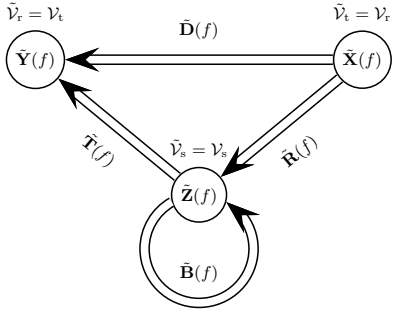
\begin{figure} \centering
\resizebox{0.9\columnwidth}{!}{
\begin{pspicture}(-1,-2)(10,4.5)
\small
\cnodeput(1.5,3){vt}{$\tilde{\vec Y}(f)$}
\cnodeput(7.5,3){vr}{$\tilde{\vec X}(f)$}
\cnodeput(4.5,0.5){vs}{$\tilde{\vec Z}(f)$}
\nput{90}{vt}{$\tilde{\mathcal V}_{\mathrm{r}} = \mathcal V_{\mathrm{t}}$}
\nput{90}{vr}{$\tilde{\mathcal V}_{\mathrm{t}} = \mathcal V_{\mathrm{r}}$}
\nput{90}{vs}{$\tilde{\mathcal V}_{\mathrm{s}} = \mathcal V_{\mathrm{s}} $}
\ncline[doubleline=true,doublesep=5\pslinewidth,arrows=<-]{vt}{vr}
\naput[nrot=:U]{$\tilde{\mat D}(f)$}
\ncline[doubleline=true,doublesep=5\pslinewidth,arrows=<-]{vt}{vs}
\nbput[nrot=:U]{$\tilde{\mat T}(f)$}
\ncline[doubleline=true,doublesep=5\pslinewidth,arrows=<-]{vs}{vr}
\nbput[nrot=:U]{$\tilde{\mat R}(f)$}
\nccircle[angleA=-180,doubleline=true,doublesep=5\pslinewidth,arrows=->]{vs}{1}
\naput{$\tilde{\mat B}(f)$}
\end{pspicture}}
  \caption{Vector signal flow graph representation of a reverse
    propagation graph. Compared to the forward graph depicted in
    Fig.~\ref{fig:blockDiagram} all edges are reversed.}
  \label{fig:blockDiagramReverse}
\end{figure} Considering Fig.~\ref{fig:blockDiagramReverse} and arguing
as in Section \ref{sec:transf-funct-recurs} yields for the reverse
channel
\begin{align}
  \label{eq:30}
\tilde{\mat H}(f) &= \tilde{\mat D} (f) + 
\tilde{ \mat R} (f) 
\left[ \mat I  - \tilde{\mat
   B} (f) \right]^{-1}\tilde{\mat T} (f)
\end{align}
Inserting for $\tilde{\mat D}(f),\tilde{\mat
  R}(f),\tilde{\mat T}(f),$ and $\tilde{\mat B}(f)$ leads to 
\begin{align}
\tilde{\mat H}(f) &= \mat D\transpose (f) + \mat T\transpose (f) \left[ \mat I  - \mat
   B\transpose (f) \right]^{-1}\mat R\transpose (f) \nonumber \\
&= \mat H\transpose
 (f),
\end{align}
where the last equality results from \eqref{eq:60}.
\end{proof}

In words Theorem~\ref{thm:recipr-prop-graphs} says that for a
propagation graph there exists a reverse graph which fulfills the
reciprocity condition. Furthermore, an example of a propagation graph fulfilling
the reciprocity condition is the reverse graph $\tilde{\mathcal G}$
obtained by reversing all edges of $\mathcal G$ while maintaining the
edge transfer functions.

\subsection{Related Recursive Scattering Models}
We provide a few examples of recursive models to assist the
reader in recognizing models which can be represented by the graphical
structure.

In \cite{Shi2007a} Shi and Nehorai consider a model for recursive
scattering between point scatterers in a homogeneous background.  The
propagation between any point in space is described by a scalar
Green's function.
 The transfer function obtained by
applying the Foldy-Lax equation can also be obtained from a
propagation graph by defining the sub-matrices of $\mat A(f)$ as
follows. The model does not include a directed term and thus $\mat
D(f) = \mat 0$. The entry of $[\mat T(f)]_{m_1n}$ is the 
Green's function from transmit vertex $m_1$ to scatterer $n$ times the
scattering coefficient of scatterer $n'$. Similarly, the entry $[\mat
R(f)]_{nm_2}$ is the Green's function from the position of
scatterer $n$ to receiver $m_2$. The entry $[\mat B(f)]_{n n'}, n\neq
n'$ is the Green's function from the position of scatterer $n$ to the
position of scatterer $n'$ times the scattering coefficient of
scatterer $n$. Since a point scatterer does not scatter back on
itself, the diagonal entries of $\mat B(f)$ are all zero.  As can be
observed from the above definitions, the assumption of homogeneous
background medium leads to the special case with $\mathcal E_{\mathrm
  d} = \emptyset$, $\mathcal E_{\mathrm t} = \mathcal V_{\mathrm
  t}\times \mathcal V_{\mathrm s}$, $\mathcal E_{\mathrm r} =
\mathcal V_{\mathrm s}\times \mathcal V_{\mathrm r}$, and $\mathcal
E_{\mathrm s} = \mathcal V_{\mathrm s}\times \mathcal V_{\mathrm s}$.

Another modeling method that can be conveniently described using
propagation graphs is (time-dependent) radiosity
\cite{Nosal2004,Hodgson2006,Siltanen2007,Muehleisen2009,Rougeron2002}.
In these methods, the surfaces of the objects in the propagation
environment are first divided in smaller patches, then the power
transfer between pairs of patches is defined in terms of the so-called
``delay-dependent form factor'', and finally the power transfer from
the transmitter to the receiver is estimated. The delay-dependent form
factor is essentially a power impulse response between patches. It
appears that for these algorithms no closed-form solution feasible for
numerical evaluation is available in the literature. Thus
\cite{Nosal2004,Hodgson2006,Siltanen2007,Muehleisen2009,Rougeron2002}
resort to iterative solutions which can be achieved after discretizing
the inter-patch propagation delays. The time-dependent radiosity
problem can be expressed in the Fourier domain in terms of a
propagation graph where each patch is represented by a scatterer,
while the entries of $\mat A(f)$ are defined according to delay
dependent form factor. Using this formulation, a closed form
expression of the channel transfer matrix appears immediately by
Theorem~\ref{thm:transf-funct-recurs} with no need for quantization of
propagation delays.


\subsection{Revisiting Existing Stochastic Radio
  Channel Models}
It is interesting to revisit existing radio channel models by means of
the just defined framework of propagation graphs. Such an effort may
reveal some structural differences between models, which are not
apparent merely from the mathematical formulation. { It is, however, a
  fact that the interpretation of a transfer function as a propagation
  graph is not unique---it is straightforward to construct different
  propagation graphs which yield the same transfer matrix. Thus
  different propagation graphs may be associated to the same radio
  channel model. In the sequel we construct graphs for two well-known
  stochastic channel models, i.e., the model by Turin \emph{et al.\ }
  \cite{turin} and the model by Saleh and Valenzuela \cite{saleh}. We
  include here only enough detail to allow for the discussion of the
  structure of the associated graphs and refer the reader to the
  original papers for the full description.

The seminal model \cite{turin} by Turin \emph{et al.}\ can be
expressed in the frequency domain as
\begin{equation}
  \label{eq:26}
  H(f) = \sum_{\ell=0}^{\infty}\alpha_\ell
  \exp(-j2\pi f \tau_\ell),
\end{equation}
where $\alpha_\ell$ is the complex gain and $\tau_\ell$ denotes the
delay of component $\ell$.  Apart from the direct component, to which
we assign index $\ell =0$, the index $\ell$ is merely a dummy index which
does not indicate any ordering of terms in \eqref{eq:26}. The set,
$\{(\tau_\ell,\alpha_\ell): \ell =1, 2, \dots\}$ is a marked Poisson
point process on $[\tau_0,\infty)$ with complex marks $\{\alpha_\ell:
\ell =1, 2, \dots \}$.

The propagation graph we seek should contain a propagation path for
each term in the sum \eqref{eq:26}. We notice that modifying the term
$\alpha_\ell \exp(-j2\pi f \tau_\ell)$ has no effect on the other
terms in \eqref{eq:26}.  This structure can be captured by
constructing a propagation graph in which modifying vertices or edges
in a propagation path leaves others unchanged.  Graph
theory offers the convenient notion of ``independent walks''
\cite{Diestel2000}: Two or more walks are called independent if none
of them contains an inner vertex of another. In particular,
propagation paths are independent if they do not traverse the same
scatterers. Thus for independent propagation paths $\ell, \ell'$,
changing path $\ell$ by changing its edge transfer functions or by
deleting edges from it, has no effect on path $\ell'$. We therefore
associate to each term in \eqref{eq:26} an independent propagation
path. Assuming for simplicity all paths to be single bounce paths this
results in the graph shown in Fig.~\ref{fig:turinSalehValenzuela}(a).

\begin{figure}
  \centering
\psset{unit=1cm}

\begin{pspicture}(-1.3,1)(4.3,3.2)
\psset{gridcolor=yellow}
\psset{subgridcolor=yellow}
\psset{linewidth=0.5pt}
\small
\rput[l]{0}(-1,3){\scriptsize a) Turin \emph{et al.}}
\rput[c]{0}(0,2){\dotnode(0,0){Tx}\nput[labelsep=0.1]{180}{Tx}{Tx}}
\rput[c]{0}(4,2){\dotnode(0,0){Rx}\nput[labelsep=0.1]{0}{Rx}{Rx}}
\rput[c]{0}(2,2.7){\dotnode(0,0){s1}}
\rput[c]{0}(2,2){\dotnode(0,0){s2}}
\rput[c]{0}(2,1.3){\dotnode(0,0){s3}}
\nput{-90}{s3}{\psdots[dotscale=0.25](0,0.1)(0,0)(0,-0.1)}
\ncline[nodesepA=0.1, nodesepB=0.1]{cc->}{Tx}{s1}\lput{:U}{}
\ncline[nodesepA=0.1, nodesepB=0.1]{cc->}{Tx}{s2}\lput{:U}{}
\ncline[nodesepA=0.1, nodesepB=0.1]{cc->}{Tx}{s3}\lput{:U}{}
\ncline[nodesepA=0.1, nodesepB=0.1]{cc->}{s1}{Rx}\lput{:U}{}
\ncline[nodesepA=0.1, nodesepB=0.1]{cc->}{s2}{Rx}\lput{:U}{}
\ncline[nodesepA=0.1, nodesepB=0.1]{cc->}{s3}{Rx}\lput{:U}{}
\end{pspicture}
\strut
\\[1.1cm]
\psset{unit=1cm}
\begin{pspicture}(-1.3,0.5)(4.3,3.5)
\psset{gridcolor=yellow}
\psset{subgridcolor=yellow}
\psset{linewidth=0.5pt}
\small
\rput[l]{0}(-1,3.7){\scriptsize b) Saleh \& Valenzuela }
\rput[c]{0}(0,2){\dotnode(0,0){Tx}\nput[labelsep=0.1]{180}{Tx}{Tx}}
\rput[c]{0}(4,2){\dotnode(0,0){Rx}\nput[labelsep=0.1]{0}{Rx}{Rx}}

\rput[c]{0}(1,3){
  \dotnode(0,0){s1}
  \dotnode(1,-0.5){s11}
  \dotnode(1,0){s12}
  \dotnode(1,0.5){s13}}
\rput[c]{0}(1,1){
  \dotnode(0,0){s2}
  \dotnode(1,-0.5){s21}
  \dotnode(1,0){s22}
  \dotnode(1,0.5){s23}}

\nput{-90}{s21}{\psdots[dotscale=0.25](0,0.1)(0,0)(0,-0.1)}
\nput{-90}{s11}{\psdots[dotscale=0.25](0,0.1)(0,0)(0,-0.1)}
\nput{-90}{s1}{\psdots[dotscale=0.25](0,0.1)(0,0)(0,-0.1)}
\nput{-90}{s2}{\psdots[dotscale=0.25](0,0.1)(0,0)(0,-0.1)}

\ncline[nodesepA=0.1, nodesepB=0.1]{cc->}{Tx}{s1}\lput{:U}{}
\ncline[nodesepA=0.1, nodesepB=0.1]{cc->}{Tx}{s2}\lput{:U}{}
\ncline[nodesepA=0.1, nodesepB=0.1]{cc->}{s1}{s11}\lput{:U}{}
\ncline[nodesepA=0.1, nodesepB=0.1]{cc->}{s1}{s12}\lput{:U}{}
\ncline[nodesepA=0.1, nodesepB=0.1]{cc->}{s1}{s13}\lput{:U}{}
\ncline[nodesepA=0.1, nodesepB=0.2]{cc->}{s11}{Rx}\lput{:U}{}
\ncline[nodesepA=0.1, nodesepB=0.2]{cc->}{s12}{Rx}\lput{:U}{}
\ncline[nodesepA=0.1, nodesepB=0.2]{cc->}{s13}{Rx}\lput{:U}{}
\ncline[nodesepA=0.1, nodesepB=0.1]{cc->}{s2}{s21}\lput{:U}{}
\ncline[nodesepA=0.1, nodesepB=0.1]{cc->}{s2}{s22}\lput{:U}{}
\ncline[nodesepA=0.1, nodesepB=0.1]{cc->}{s2}{s23}\lput{:U}{}
\ncline[nodesepA=0.1, nodesepB=0.2]{cc->}{s21}{Rx}\lput{:U}{}
\ncline[nodesepA=0.1, nodesepB=0.2]{cc->}{s22}{Rx}\lput{:U}{}
\ncline[nodesepA=0.1, nodesepB=0.2]{cc->}{s23}{Rx}\lput{:U}{}

\end{pspicture}

\caption{Propagation graph representations of: a) a realization of the
  model by Turin \emph{et al.} \cite{turin}; and  b) a realization of the
  Saleh-Valenzuela model \cite{saleh}.
}
  \label{fig:turinSalehValenzuela}
\end{figure}
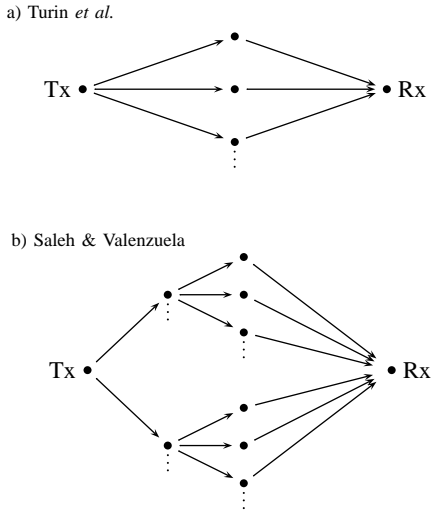

The model by Saleh and Valenzuela \cite{saleh} can be formulated as a
second-order Turin model:
\begin{equation}
  \label{eq:40}
  H(f) = \sum_{\ell=0}^{\infty} 
  \exp(-j2\pi f \tau_\ell)\sum_{\ell'=0}^{\infty} \alpha_{\ell\ell'} \exp(-j2\pi f \tau_{\ell\ell'}),
\end{equation}
where $\ell$ is a cluster index and $\ell'$ is a index
for the components within a cluster. The set $\{\tau_\ell: \ell =
1,2,\dots\}$ is a Poisson process on $[\tau_0,\infty)$ and 
conditioned on the cluster delays $\{\tau_\ell; \ell = 0,1,\dots \}$,
the family of sets $\{(\alpha_{\ell\ell'},\tau_{\ell\ell'}): \ell' =
1,2,\dots\}$, $\ell = 0,1,\dots,$ are independent marked Poisson
processes. Using a similar argument as for the Turin model,  leads to the
graph depicted in Fig.~\ref{fig:turinSalehValenzuela}(b).
%
The transmitter is connected to a set of
super-ordinate scatterers, each scatterer corresponding to a
``cluster''. These cluster-scatters are then connected to the
receiver via independent single-bounce paths passing
through sub-ordinate scatterers.  In this graph, deleting 
cluster scatterer $\ell$ makes a whole signal cluster $\ell$ disappear
in \eqref{eq:40}. Similarly, removing the sub-ordinate scatterer
$\ell'$ the component with index $\ell\ell'$ in \eqref{eq:40}
vanishes. 


We end the discussion of the Saleh-Valenzuela model by mentioning that
many equivalent graphical interpretations may be given. As an example,
one may consider a graph structure as that of the reverse graph of
Fig.~\ref{fig:turinSalehValenzuela}(b). In such a structure, the
transmitter is directly connected to the sub-ordinate scatterers
whereas the receiver is connected to the clusters. Indeed by reversing
any number of clusters we obtain an equivalent propagation
graph. These graphs share the common property that they contain no
cycles.

\section{Example: Stochastic Model for In-Room Channel} 
The concept of propagation graph introduced until now can be used for
describing a broad range of channel models. In this section we apply
these general results to a specific example scenario where scatterer
interactions are considered to be non-dispersive in delay. { This
  resembles the case where all wave interactions can be considered as
  specular reflections, or considering points scatterers.  The intent
  is to exemplify that the experimentally observed avalanche effect
  and diffuse tail can be explained using only scatterer interactions
  that by themselves yield no delay dispersion.}  We specify a method feasible for generating such a
graph in Monte Carlo simulations. The model discussed in this example
is a variant of the model proposed in
\cite{Pedersen2007,Pedersen2006}.

\subsection{Weighted Adjacency Matrix}
We define the weighted adjacency matrix according to a geometric model
of the environment. We consider a scenario with a single transmitter,
a single receiver, and $N_{\mathrm s}$ scatterers, i.e., the vertex set reads
$\mathcal V = \mathcal V_\mathrm{t} \cup \mathcal V_\mathrm{r} \cup
\mathcal V_\mathrm{s}$ with $ \mathcal V_\mathrm{t} = \{\mathrm{Tx}\},
\mathcal V_\mathrm{r} = \{\mathrm{Rx}\},$ and $ \mathcal V_\mathrm{s}
= \{\mathrm{S}1,\dots,\mathrm{S}N_{\mathrm s}\}. $ To each vertex $v\in \mathcal
V$ we assign a displacement vector $\vec r_v \in\mathbb{R}^3$ with
respect to a coordinate system with arbitrary origin.  To edge
$e=(v,v')$ we associate the Euclidean distance $d_e = \|\vec r_v - \vec r_{v'}
\|$, the gain $g_e$, the phase $\phi_e$, and the propagation delay
$\tau_e = d_e/c$, where $c$ is the speed of light.  The edge transfer
functions are defined as
\begin{equation}
  \label{eq:14}
  A_e(f) = 
  \begin{cases}
g_e(f) \exp j(\phi_e -2\pi \tau_e f); & e\in \mathcal E\\
0; &     e\not\in \mathcal E.\\
  \end{cases}
\end{equation}    
The edge gains $\{g_e(f)\}$ are defined according to 
\begin{equation}
  \label{eq:16}
  g^2_e(f) =
  \begin{cases}
    \frac{1}{(4\pi f \tau_{e})^2};& e\in \mathcal E_{\mathrm d} \\
        \frac{1}{4\pi f \mu(\mathcal E_{\mathrm t})}\cdot
        \frac{\tau_e^{-2}}{S(\mathcal E_{\mathrm t})}; &
 e\in\mathcal E_{\mathrm t} \\
    \frac{1}{4\pi f \mu(\mathcal E_{\mathrm r})}\cdot 
    \frac{\tau_e^{-2}}{S(\mathcal E_{\mathrm r})};  & e\in\mathcal E_{\mathrm r} \\
 \frac{g^2}{\mathrm{odi}(e)^2};   & e\in
\mathcal E_{\mathrm s}
  \end{cases}
\end{equation}
where $\mathrm{odi}(e)$ denotes the number of edges from
$\mathrm{init}(e)$ to other scatterers and for any $\mathcal E'\subseteq
\mathcal E$ 
\begin{equation}
  \label{eq:37}
  \mu(\mathcal E') = \frac{1}{|\mathcal E'|}\sum_{e\in\mathcal E'}
  \tau_e  \qquad \text{and} \qquad  S(\mathcal E') =
  \sum_{e\in\mathcal E'} \tau_e^{-2},
\end{equation} 
with $|\cdot |$ denoting cardinality.  The weight of the direct edge
is selected according to the Friis equation \cite{Friis1946} assuming
isotropic antennas at both ends. The weights of edges in $\mathcal
E_{\mathrm t}$ and $\mathcal E_{\mathrm r}$ also account for the
antenna characteristics. They are computed at the average distance to
avoid signal amplification when scatterers are close to a transmitter
or receiver, namely when the far-field assumption is
invalid. The value of inter-scatterer gain $g$ is for simplicity assumed
fixed. 


\subsection{Stochastic Generation of Propagation Graphs}
\label{sec:stoch-gener-prop}
We now define a stochastic model of the sets $\{\vec r_v\}$, $\mathcal
E$, and $\{\phi_e\}$ as well as a procedure to compute 
the corresponding transfer function and impulse response.
The vertex positions are assumed to reside in a bounded region $\mathcal
R\subset \mathbb R^3$ corresponding to the region of interest.  The
transmitter and receiver positions are assumed to be fixed, while the
positions of the $N_{\mathrm s}$ scatterers $\{\vec r_v: v\in \mathcal
V_s \}$ are drawn independently from a uniform distribution on $\mathcal R$.

Edges are drawn independently such that a vertex pair $e\in \mathcal
V\times\mathcal V$ is in the edge set $\mathcal E$ with probability $P_{e} =
\mathrm{Pr}[e\in \mathcal E] $ defined as
\begin{equation}
\label{eq:5}
    P_{e} = 
    \begin{cases}
      P_{\mathrm{dir}}, & e = (\mathrm{Tx},\mathrm{Rx})  \\
      0, &  \mathrm{term}(e)  = \mathrm{Tx} \\
      0, &  \mathrm{init}(e)  = \mathrm{Rx} \\
      0, &  \mathrm{init}(e) =\mathrm{term}(e) \\
      P_\mathrm{vis} & \text{otherwise}
    \end{cases}.
\end{equation}
The first case of \eqref{eq:5} controls the occurrence of a direct
component. If $P_\mathrm{dir}$ is zero, $ D(f)$ is
zero with probability one. If $P_\mathrm{dir}$ is unity, the direct
term $D(f)$ is non-zero with probability one.  The second and
third cases of \eqref{eq:5} exclude ingoing edges to the transmitter
and outgoing edges from the receiver. Thus the generated graphs will
have the structure defined in Section~\ref{subsec:propagation-graphs}.
The fourth case of \eqref{eq:5} excludes the occurrence of loops in
the graphs. This is sensible as a specular scatterer cannot scatter a
signal back to itself. A consequence of this choice is that any
realization of the graph is loopless and therefore $\mat A(f)$ has
zeros along its main diagonal.  The last case of \eqref{eq:5} assigns
a constant probability $P_{\mathrm {vis}}$ of the occurrence of edges
from $\mathcal V_t$ to $\mathcal V_s$, from $\mathcal V_s$ to
$\mathcal V_s$ and from $\mathcal V_s$ to $\mathcal V_r$.


Finally, the phases $\{\phi_e:e\in \mathcal E\}$ are drawn
independently from a uniform distribution on  $[0;2\pi)$.

Given the deterministic values of parameters $\mathcal R$, $\vec r_{\mathrm {Tx}}$, $\vec
r_{\mathrm {Rx}}$, $N_{\mathrm s}$, $P_{\mathrm {dir}}$, $P_{\mathrm {vis}}$ and
$g$, realizations of the (partial) transfer function $H_{K:L}(f)$ and
corresponding (partial) impulse response $h_{K:L}(\tau)$ can now be
generated for a preselected frequency range
$[f_{\mathrm{min}},f_{\mathrm{max}}]$, using the algorithm stated in
Fig.~\ref{fig:algorithm} and Appendix.
\begin{figure}
  \centering
\fbox{
  \begin{minipage}{0.95\columnwidth}
\small{
\begin{enumerate}
\item Draw   $\vec r_v,\ v\in \mathcal
   V_{\mathrm{s}}$ uniformly on $\mathcal R$\\[-1.5ex] 
 \item Generate  $\mathcal{E}$ according to \eqref{eq:5}
\\[-1.5ex] 
 \item Draw independent  phases $\{\phi_e: e\in \mathcal E \}$
   uniformly on $[0,2\pi)$\\[-1.5ex] 
 \item Compute $\mat A(f)$ within the frequency bandwidth using \eqref{eq:14}\\[-1.5ex] 
 \item IF spectral radius of $\mat B(f)$ is larger than unity for some frequency within
   the bandwidth THEN  GOTO step 1\\[-1.5ex] 
 \item Estimate $H_{K:L}(f)$ and $h_{K:L}(\tau)$  as described
   in Appendix~\ref{sec:numer-comp-transf} \\[-1.5ex] 
\end{enumerate}}
\end{minipage}
}
  \caption{Algorithm for generating full or partial transfer functions
    and  impulse responses for a preselected bandwidth. }
  \label{fig:algorithm}
\end{figure}

\subsection{Numerical Experiments}
The effect of the recursive scattering phenomenon can now be
illustrated by numerical experiments. The parameter settings given in
Table~\ref{tab:simSettings} are selected to mimic the experimental
setup of \cite{Kunisch2003} used to acquire the measurements reported
in Fig.~\ref{fig:kunischExample}. The room size and positions of the
transmitter and receiver are chosen as in \cite{Kunisch2003}. We
consider the case where direct propagation occurs and set
$P_{\mathrm{dir}}$ to unity. The probability of visibility
$P_{\mathrm{ vis}}$ and the number of scatterers $N_{\mathrm s}$ are chosen to
mimic the observed avalanche effect.  The value of $g$ is set to match
the tail slope $\rho \approx -0.4\, \deci\bel\per\nano\second$ of the
delay-power spectrum depicted in Fig.~\ref{fig:kunischExample}. The
value of $ g$ can be related to the slope $\rho$ of the log delay-power spectrum via the approximation $\rho \approx 20\log_{10}(g) /
\mu(\mathcal E_{\mathrm s})$. This approximation arises by considering
the power balance for a scatterer:  assuming  the signal components
arriving at the scatterer to be statistically independent, neglecting
the probability of scatterers with outdegree zero, and approximating
the delays of edges in $\mathcal E_{\mathrm s}$ by the average $
\mu(\mathcal E_{\mathrm s})$ defined in \eqref{eq:37}.

\begin{table}

 \centering
  \caption{Parameter Settings for the Numerical Examples}
  \label{tab:simSettings}
\scriptsize
\begin{tabular}[c]{lcc}

\toprule
Parameters & Symbol & Values\\ \midrule

Room size& $\mathcal{R}$ &   $[0,5]\times[0,5]\times[0,2.6] \,\cubic\meter$
\\
Transmitter position& $\vec r_{\mathrm{Tx}}$& $ [1.78,\, 1.0,\, 1.5]^\mathrm{T}\,\meter$\\
Receiver position& $\vec r_{\mathrm{Rx}}$& $ [4.18,\,  4.0,\, 1.5]^\mathrm{T}\,\meter$\\
Number of scatterers& $N_{\mathrm s}$ & 10\\
Tail slope & $\rho$ &--0.4\,\deci\bel\per\nano\second \\
Prob. of visibility& $P_\mathrm{vis}$& 0.8\\
Prob. of direct propagation& $P_\mathrm{dir}$& 1\\
Speed of light& $c$ & 3\cdot10\textsuperscript 8 \meter\per\second\\
Transmit signal& $X[m]$& Unit power Hann pulse \\
Number of frequency samples &$M $ & 8192\\ 
\bottomrule
  \end{tabular}
\end{table}

\begin{figure}
  \centering
\resizebox{0.8\columnwidth}{!}{
  \input{./figures/channelResponse.tex}
  \includegraphics{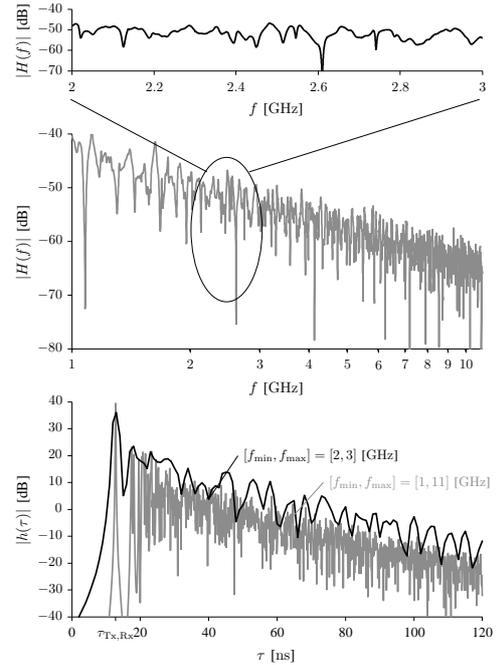}}
\caption{Channel response for a specific realization of the
  propagation graph. Top: Transfer function in dB
  $(20\log_{10}|H(f)|)$. Bottom: Impulse responses in dB
  $(20\log_{10}|h(\tau)|)$ computed for two frequency bandwidths.
}
  \label{fig:channelResponse}
\end{figure}

Fig.~\ref{fig:channelResponse} shows the amplitude of a single
realization of the transfer function. Overall, the squared amplitude
of the transfer function decays as $f^{-2}$ due to the definition of
$\{g_e(f)\}$. Furthermore, the transfer function exhibits fast fading
over the considered frequency band. The lower panel of
Fig.~\ref{fig:channelResponse} reports the corresponding impulse
response for two different signal bandwidths.  Both impulse responses
exhibit an avalanche effect as well as a diffuse tail 
of which the power decays nearly exponentially with $\rho \approx -0.4\,
\deci\bel\per\nano\second$.  As anticipated, the transition to the
diffuse tail is most visible in the response obtained with the larger
bandwidth.

\begin{figure*}
    \centering 
    \resizebox{0.8\textwidth}{!}{
    \input{./figures/partialResponseOverview.tex}
    \includegraphics{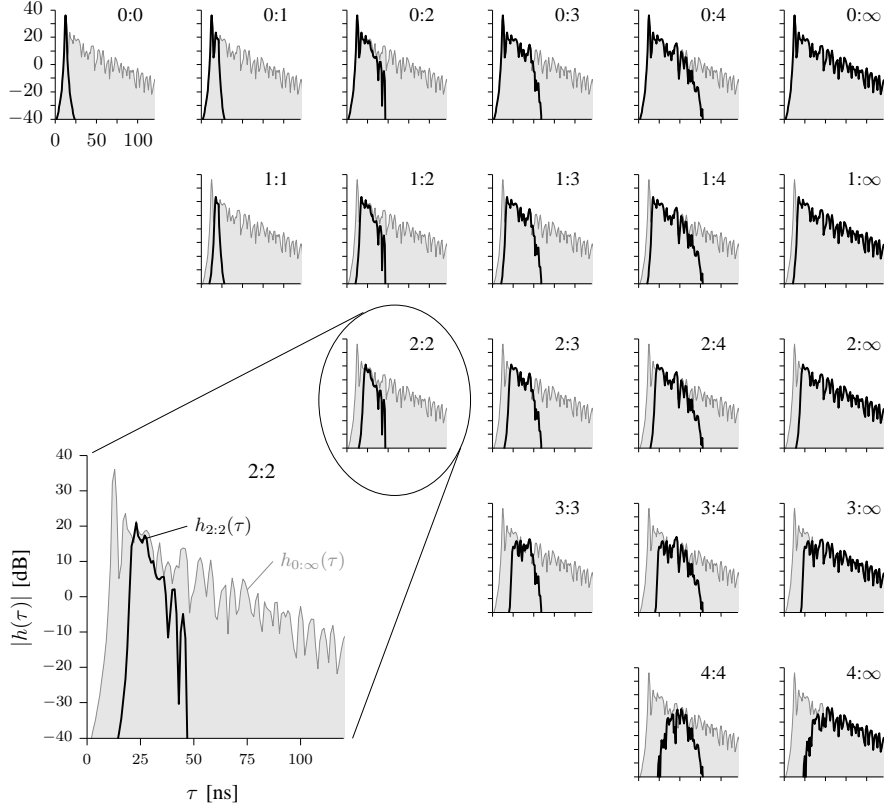}
    }
    \caption{Partial responses obtained for  one graph realization for
      the  bandwidth      $ [2,3]\,\giga\hertz$. The
      $K:L$ settings are indicated in each miniature. 
      Top row:
      $K$-bounce approximations. Right-most column: error terms
      resulting from $(K-1)$-bounce approximations. Main diagonal
      $(K=L)$: $K$-bounce contributions. }
    \label{fig:partialResponseOverview}
\end{figure*}

The build up of the impulse response can be examined via the partial
impulse responses given in Fig.~\ref{fig:partialResponseOverview}.
Inspection of the partial responses when $K=L$ reveals that the early
part of the tail is due to signal components with a low $K$ while the
late part is dominated by higher-order signal components. It can also be
noticed that as $K$ increases, the delay at which the maximum of the
$K$-bounce partial response occurs and the spread of this response are
increasing.

\begin{figure}
    \centering
    \resizebox{\columnwidth}{!}{
    \input{figures/delayPowerSpectra.tex}
    \includegraphics{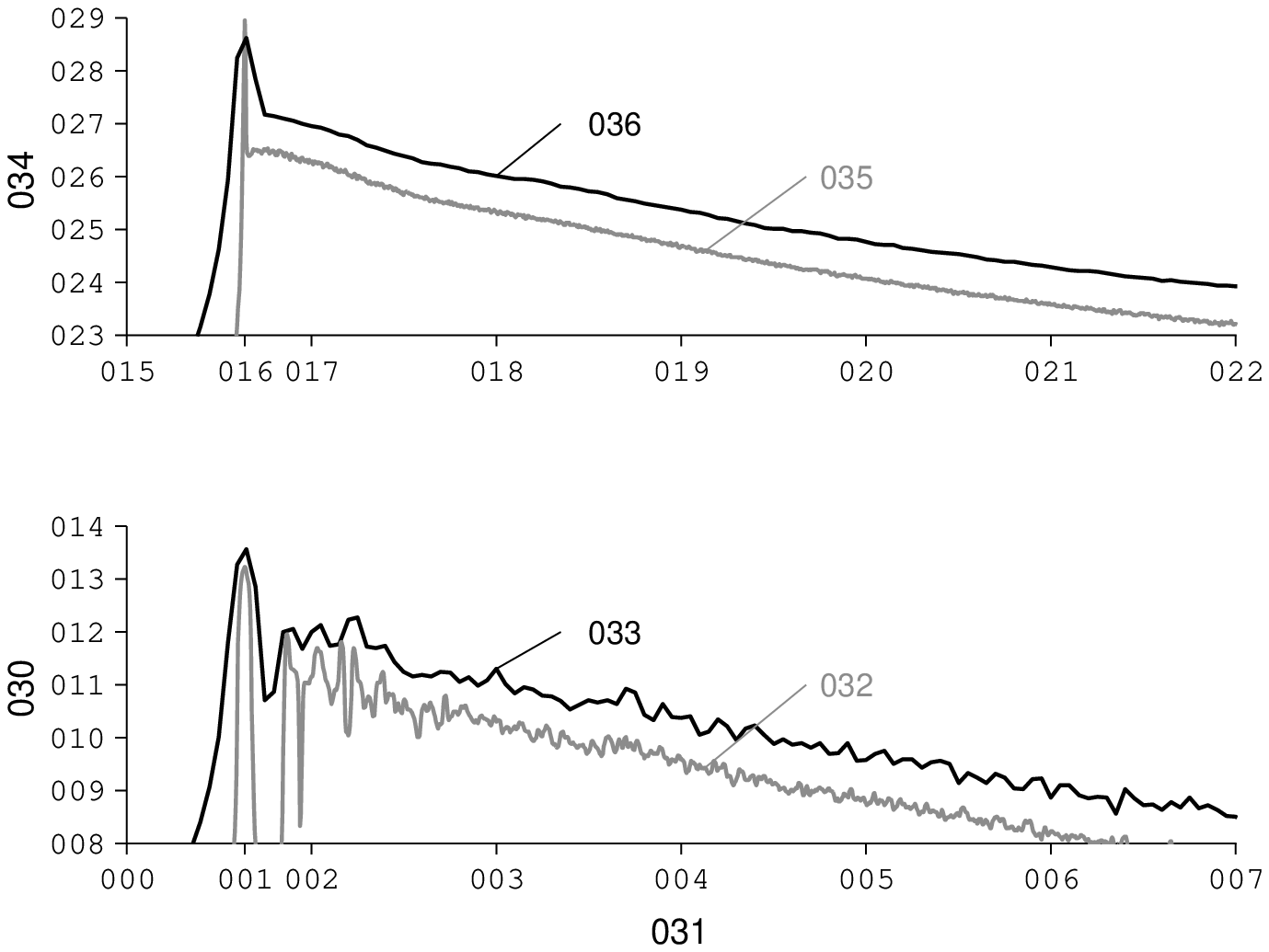}}
  \caption{Simulated delay-power spectra. 
    Top panel: Ensemble average of 1000 Monte Carlo runs. Bottom
    panel: Spatial average of a single graph realization assuming the
    same grid as in Fig.~\ref{fig:kunischExample}, i.e.,
    900 receiver positions on a 30$\times$30 horizontal square grid
    with $1\times1\,\centi\meter\square{}$ mesh centered at position
    $\vec r_{\mathrm{Rx}}$ given in Table~\ref{tab:simSettings}.}
    \label{fig:delayPowerSpectra}
\end{figure}

Fig.~\ref{fig:delayPowerSpectra} shows two types of delay-power
spectra. The upper panel shows the ensemble average of squared
amplitudes of 1000 independently drawn propagation graphs for the two
signal bandwidths also considered in Fig.~\ref{fig:channelResponse}.
Both spectra exhibit the same trend: A clearly visible peak due to the
direct signal is followed by a tail with nearly exponential power decay. As
expected, the first peak is wider for the case with 1\,\giga\hertz{}
bandwidth than for the case with 10\,\giga\hertz{} bandwidth. The
tails differ by approximately 7~dB. This shift arises due to the
$f^{-2}$ trend of the transfer function resulting in a higher received
power for the lower frequencies considered in the 1\,\giga\hertz{}
bandwidth case.

The bottom panel shows spatially averaged delay-power spectra obtained
for one particular realization of the propagation graph.  The
simulated spatial averaged delay-power spectra exhibit the avalanche
effect similar to the one observed in Fig.~\ref{fig:kunischExample}.
For the 10\,\giga\hertz{} bandwidth case the power level of
the diffuse tails agrees remarkably well with
the measurement in Fig.~\ref{fig:kunischExample}. The modest deviation of
about 3~\deci\bel{} can be attributed to antenna losses in the measurement.
{
\section{Conclusions}
\label{sec:conclusion}
The outset for this work was the observation that the measured impulse
responses and delay-power spectra of in-room channels exhibit an
avalanche effect: Multipath components appear at increasing rate and
gradually merge into a diffuse tail with an exponential decay of
power.  We considered the question whether the avalanche effect is due
to recursive scattering. To this end we modelled the propagation
environment as a graph in which vertices represent transmitters,
receivers, and scatterers and edges represent propagation conditions
between vertices. From this general structure the graph's full and
partial frequency transfer matrices were derived in closed form. These
expressions can, by specifying the edge transfer functions, be
directly used to perform numerical simulations.

We considered as an example a graph-based stochastic model where all
interactions are non-dispersive in delay in a scenario similar to an
experimentally investigated scenario where the avalanche effect was
observed. The impulse responses generated from the model also exhibit an
avalanche effect. This numerical experiment lead to the observation
that the diffuse tail can be generated even when scatterer
interactions are non-dispersive in delay. Furthermore, the exponential
decay of the delay-power spectra can be explained by the presence of
recursive scattering. As illustrated by the simulation results the
proposed model, in contrast to existing models which treat dominant
and diffuse components separately, provides a unified account by
reproducing the avalanche effect.  }

\appendix
\label{sec:numer-comp-transf}
The transfer function $\mat H(f)$ and impulse response $ \mat h (\tau)$ can be
estimated as follows:
\begin{enumerate}
\item Compute $M$ samples of the transfer matrix within the
  bandwidth $[f_{\mathrm  {min}},f_{\mathrm  {max}}] $
  \begin{equation}
    \label{eq:28}
\mat H[m] = \mat H(f_{\mathrm{min}}+
  m \Delta_f),\quad
  m =0 , 1, \dots  M-1,
  \end{equation}
  where $\Delta_f = {(f_{\mathrm {max}}-f_{\mathrm
      {min}})}/({M-1})$ and $\mat H(\cdot)$ is obtained using
  Theorem~\ref{thm:transf-funct-recurs}.
\item Estimate the received signal $\mat y(\tau)$ via the inverse
  discrete Fourier transform
\begin{multline}
  \label{eq:13}
  \mat  y(i\Delta_\tau) = \Delta_f\!\sum_{m=0}^{M-1} \mat
  H[m]
 \mat X[m] \exp(j2\pi i m/M), \\ i =0,\dots M-1, 
\nonumber
\end{multline}
where $\mat X [m] = \mat X(f_{\mathrm{min}}+m \Delta_f ),
  m =0 , 1, \dots  M-1 $ and 
$\Delta_\tau = 1/(f_{\mathrm {max}}-f_{\mathrm {min}})$.
\end{enumerate}
The impulse response can be estimated by letting $\mat X[f]$ be a
window function of unit power
\begin{equation}
  \label{eq:33}
\int^{f_{\mathrm {max}}}_{f_{\mathrm {min}}} |\mat X(f)|^2 \mathrm{d} f
\approx \sum_{m=0}^{M}| \mat X[m]]|^2\Delta_f = 1,
\end{equation}
where $\mat X[m] = \mat X(f_{\mathrm{min}}+m \Delta_f)$. The
window function must be chosen such that its inverse Fourier transform
exhibits a narrow main-lobe and sufficiently low side-lobes; $\mat
y(\tau)$ is then regarded as a good approximation of the impulse
response of the channel and by abuse of notation denoted by $\mat
h(\tau)$. Samples of the partial transfer matrix are obtained by
replacing $\mat H(\cdot ) $ by $\mat H_{K:L}(\cdot)$ in
\eqref{eq:28}. The corresponding received partial impulse response is
denoted by $\mat h_{K:L}(\tau)$.

\bibliographystyle{IEEEtran}
\bibliography{../../../referenceDataBase/referencedatabase}

\end{document}